%% file: main.tex
\documentclass[12pt, draftclsnofoot, onecolumn]{IEEEtran}

\IEEEoverridecommandlockouts
\usepackage{cite}
\usepackage{amsthm,amsmath,amssymb,amsfonts}
\usepackage{algorithmic}
\usepackage{graphicx}
\usepackage{textcomp}
\usepackage{xcolor}
\def\BibTeX{{\rm B\kern-.05em{\sc i\kern-.025em b}\kern-.08em
		T\kern-.1667em\lower.7ex\hbox{E}\kern-.125emX}}

\usepackage{url}
\usepackage[T1]{fontenc}

\usepackage{soulutf8}
\usepackage{tikz}
\usepackage{pgfplots}
\usepackage{bm}
\usepackage{cite}
\usepackage{acro} \acsetup{format/foreign = {\emph}}
\usepackage[capitalise]{cleveref}

\newtheorem{proposition}{Proposition}

\newtheorem{corollary}{Corollary}

\newtheorem{definition}{Definition}
\newtheorem{lemma}{Lemma}

\usepackage{soul}
\soulregister\ac7

\DeclareRobustCommand{\stirling}{\genfrac\{\}{0pt}{}}
\input{acronyms.tex}
\graphicspath{{Figures/}}
\pgfplotsset{compat=1.17}

\begin{document}

%
\title{A novel semantic-functional approach for multiuser event-trigger communication}

\author{Pedro E. Gória Silva,
Pl\'inio S. Dester,
Harun Siljak, \textit{Senior Member}, \textit{IEEE},
Nicola Marchetti, \textit{Senior Member}, \textit{IEEE},
Pedro H. J. Nardelli,  \textit{Senior Member},
Rausley A. A. de Souza, \textit{Senior Member}, \textit{IEEE}\thanks{Pedro Gória Silva and Pedro H. J. Nardelli  are with Lappeenranta--Lahti University of Technology, Finland, (email: pedro.goria.silva@lut.fi, pedro.nardelli@lut.fi). 
P. H. J. Nardelli  is also with University of Oulu, Finland. 
Pl\'inio S. Dester is with State University of Campinas (Unicamp), Brazil (email: pliniodester@gmail.com).
Harun Siljak and Nicola Marchetti are with Trinity College Dublin, Ireland (email: harun.siljak@tcd.ie, nicola.marchetti@tcd.ie).
Rausley A. A. de Souza is with National Institute of Telecommunications (Inatel), Santa Rita do Sapucaí 37540-000, Brazil. E-mail: rausley@inatel.br.
}}

\maketitle
\begin{abstract}
%
%

This work introduces a new perspective for physical media sharing in multiuser communication by jointly considering (i) the meaning of the transmitted message and (ii) its function at the end user. 
Specifically, we have defined a scenario where multiple users (sensors) are continuously transmitting their own states concerning a predetermined event. 
On the receiver side there is an alarm monitoring system, whose function is to decide whether such a predetermined event has happened in a certain time period and, if yes, in which user. 
The media access control protocol proposed constitutes an alternative approach to the conventional physical layer methods, because the receiver does not decode the received waveform directly; rather, the relative position of the absence or presence of energy within a multidimensional resource space carries the (semantic) information. 
The protocol introduced here provides high efficiency in multiuser networks that operate with event-triggered sampling by enabling a constructive reconstruction of transmission collisions. 
We have demonstrated that the proposed method leads to a better event transmission efficiency than conventional methods like \acs{TDMA} and slotted ALOHA. 
Remarkably, the proposed method achieves 100\% efficiency and 0\% error probability in almost all the studied cases, while consistently outperforming \acs{TDMA} and slotted ALOHA.

\end{abstract}

\acresetall

\section{Introduction} \label{sec: Intro}
For many years, \acp{WSN} have been considered building blocks of new applications, e.g., in cities and industrial plants \cite{ALSKAIF2017141}.
Such cases require data networks that could run context-aware monitoring of different processes related to the infrastructure where many sensors are spread over a relatively small area, periodically reporting their readings to a central node. 
In this context, transmissions can be multimedia bursts, such as on-demand video streams, periodic samples of a signal, and flags for specific events~\cite{Chen6757189}.
%
%

Under the umbrella of \acp{WSN}, a distinct class of sensor nodes that operate by event-driven (or event-triggered) data acquisition has emerged in recent years~\cite{MarekBook}.
With the goal of reducing the amount of data acquired and transmitted by sensor nodes, the event-driven approach---when properly designed---provides an efficient way to acquire continuously sensed data \cite{de2022event}. 
In contrast to traditional periodic sampling techniques, the event-driven approach corresponds to nonperiodic sampling. 
Roughly speaking, event-driven data acquisition is based on the fact that the event-sampled signal can be properly reconstructed (in terms of specific error functions) as long as one is aware of the occurrence of a specific event.
Such a method is interesting in cases where the signal to be sampled is more impulsive and such sudden changes are important, or when atypical or unexpected behavior of a given signal needs to be identified. 
For example, an event could be an overtemperature measured by a particular sensor or electricity metering.

Similarly to event-based signal processing, event-driven communication has potential advantages like reducing the number of unnecessary transmissions by sensors, focusing only on events that are more informative \cite{Miskowicz6010049,Santos8836592,Castro8316955}.
A crucial issue for such nonperiodic transmissions relates to the synchronization of transmissions and packet identification.
In this regard, the literature presents some techniques to perform this synchronization.
For example, the IEEE 802.11~\cite{IEEE802} standard defines a known binary sequence sent before data transmission to synchronize the transmitter and the receiver.
This synchronization solution imposes an extension on the transmitted packet, and, therefore, it may lead to an overly large header. 
For instance, given an 8-bit accuracy in the sensor reading, the synchronization message used by IEEE 802.11~\cite{IEEE802} would be eight times the payload size.
Another solution is to establish precisely synchronized clocks in the receiver and the transmitter. 
In this way, the receiver can experience the exact start and end of the bit.
This solution is widely adopted in systems based on \ac{OFDM}~\cite{Fazel2008}. 
The physical layer solution to be proposed here resembles a temporal synchronization \ac{OFDM} model.

Another fundamental issue for multiuser networks is physical media sharing. 
Specifically, we are interested in the communication from a potentially large number of sensors to a central node (gateway) forming a many-to-one topology where multiple users share the same physical medium to transmit to the gateway. 
To solve this issue, a wide range of \ac{MAC} protocols has been proposed as a way to deal with possible collisions by controlling which nodes can access the network shared resources, and/or allowing for retransmissions  \cite{popovski2020wireless,dester2021delay}. 

The simplest solution is to use random access protocols like ALOHA, where nodes transmit whenever they have a packet.
The downside is that the network performance in terms of throughput is maximized with a relatively high number of collisions. 
As expected, whenever a collision happens, the network resources, including energy, are wasted (although the network performance is optimized).
To mitigate this issue, some \ac{MAC} protocols have been designed to establish collision-free communication, which can only be achieved through overheads, centralized resource allocation, and/or contention-based protocols.
Those issues are well-known in the literature as presented in \cite{popovski2020wireless}.

Some new concepts and related technologies have emerged over the last few years, specially in upper-layer network control as, for example, named-data networks \cite{6231276}, semantic-plane protocols \cite{Popovski2020}, zero-touch networks \cite{bega2020aztec}, and software-defined networks \cite{li2016random}.
Grant-free access in cellular networks focusing on machine-type communication has also been studied in \cite{riolo2021modeling,jiang2021massive}.
In this context, a very recent and promising approach is \textit{semantics-empowered communication}, whose ambitious aim is to change the widespread ``agnostic'' paradigm of communication engineering to allow a timely generation and provision of information to the correct processing point~\cite{9475174}.
In this approach, the data are quantitatively measured in terms of their importance, and the reading and transmission of data are then regulated by this metric.
The results demonstrated a significant reduction in the traffic load among other improvements\footnote{There are a few other recent papers on semantic, or goal-oriented, communication which are not yet published, e.g., \cite{Agheli,Qin,Stamatakis}; while we are aware of them, as will become clear later on, our approach differs significantly from them.}.

Despite the unquestionable importance of the above-mentioned approaches, they diverge considerably from the concept to be proposed here.
In fact, our objective is to propose a radical change in the established way of designing wireless communication systems by incorporating both the semantics and function of data to be transmitted, in some sense modifying the well-accepted layered network models.
We consider our approach disruptive for two main reasons: (i) collisions during transmission do not make the transfer of information unfeasible; and (ii) the proposed approach takes advantage of semantic-functional knowledge about the data to design the physical layer.
In other words, we study a communication system where events with semantic value are defined, which will be transmitted to another point where this information will be used to perform a specific task.
In this sense, the acquired data have a meaning and a function.
%
%

Next, we summarize our main contributions:
\begin{itemize}
    \item We propose a \ac{SFC} system for sensor networks using predefined events to be employed in alarm detection.
    \item For the proposed \ac{SFC} system, we establish a new way of dealing with collisions by constructing random maps.
    \item New expressions of performance analysis (error probability and transmission efficiency) are derived for the proposed \ac{SFC}, \ac{TDMA}, and slotted ALOHA following the proposed model.
    \item We demonstrate by numerical results that the proposed \ac{SFC} outperforms slotted ALOHA- and \ac{TDMA}-based systems in all studied scenarios.
    %
    \item We show that the proposed solution achieves a transmission efficiency of 100\% and an average link error probability of 0\% in most of the studied cases, which indicates its extremely effective capability of allocating time and frequency resources.
\end{itemize}

The remainder of the paper is organized as follows:
\cref{sec:preli} introduces the scenario to be studied compared with other existing solutions.
The network topology and the proposed physical layer are presented in detail in \cref{sec:sys}.
\cref{sec:mapping} describes how to construct maps to transmit events related to different sensors, also showing how \ac{TDMA} and slotted ALOHA systems can be designed for comparison purposes.
The performance analysis and numerical results are provided in \cref{sec: Perform} and \cref{sec: Numerical}, respectively.
\cref{sec: Conclusion} concludes the paper.

\section{Preliminaries}
\label{sec:preli}
This section introduces the scenario to be studied in order to explain the proposed solution designed for semantic-functional and  event-triggered  communication systems. 
Our aim here is to clarify the novelty and benefits of such a solution compared with established modulation and random access control methods.
\subsection{Event-Triggered Sampling and Communication}

\begin{figure*}
    \centering
    \includegraphics[width=1\linewidth]{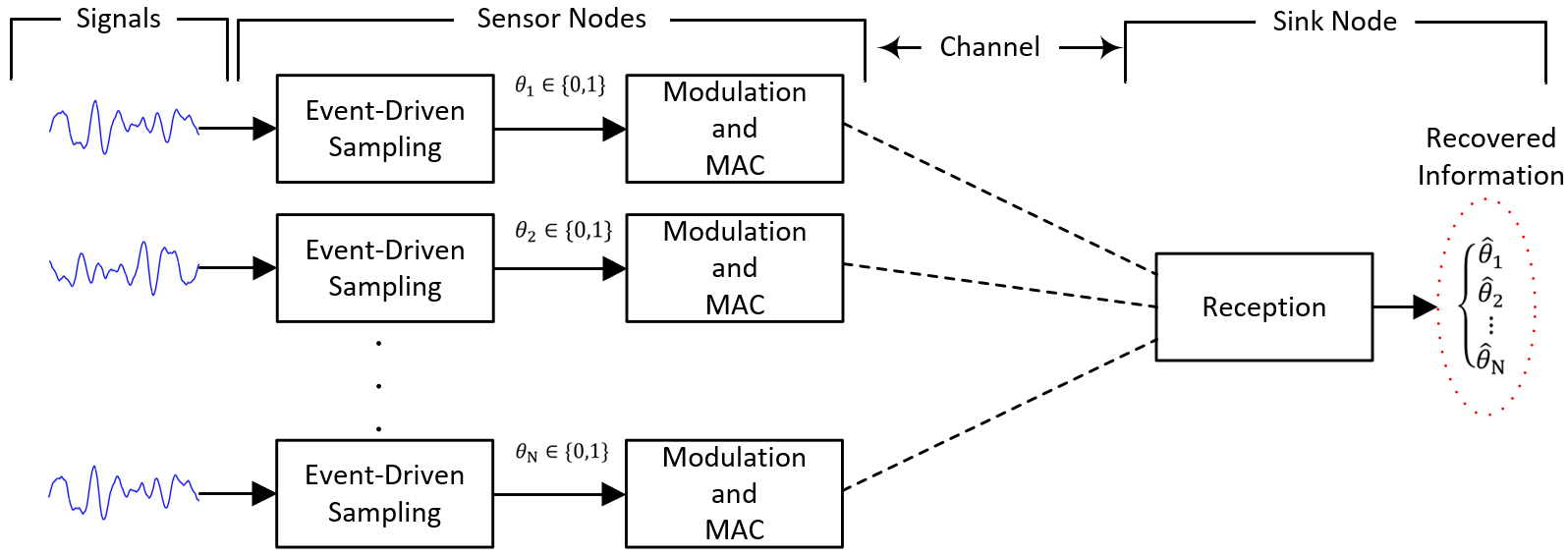}
    \caption{Block diagram of the main components of the system.}  
    \label{fig: scenario}
\end{figure*}

The first step is to present a schematic of the proposed event-triggered sampling and communication with the essential elements of the model as illustrated in  Fig. \ref{fig: scenario}. 
From left to right in the figure, we have: (i) different signals obtained from the monitored physical processes (e.g., the temperature in different positions of an industrial plant); (ii) the sensor nodes composed of data acquisition and transmission stages; (iii) a multiuser communication channel; and (iv) the sink node that needs to recover the transmitted information.
It is worth noting that in the proposed scenario, the goal of the sink is not to faithfully recover the monitored physical signals, but rather to flag whether a predetermined event at a given sensor has happened.
Fig. \ref{fig: exp_signal} depicts the aforementioned situation, where  the event related to a signal is a ``threshold crossing with positive derivative.'' 
\begin{figure}[ht!]
    \centering
    \includegraphics[width=1\linewidth]{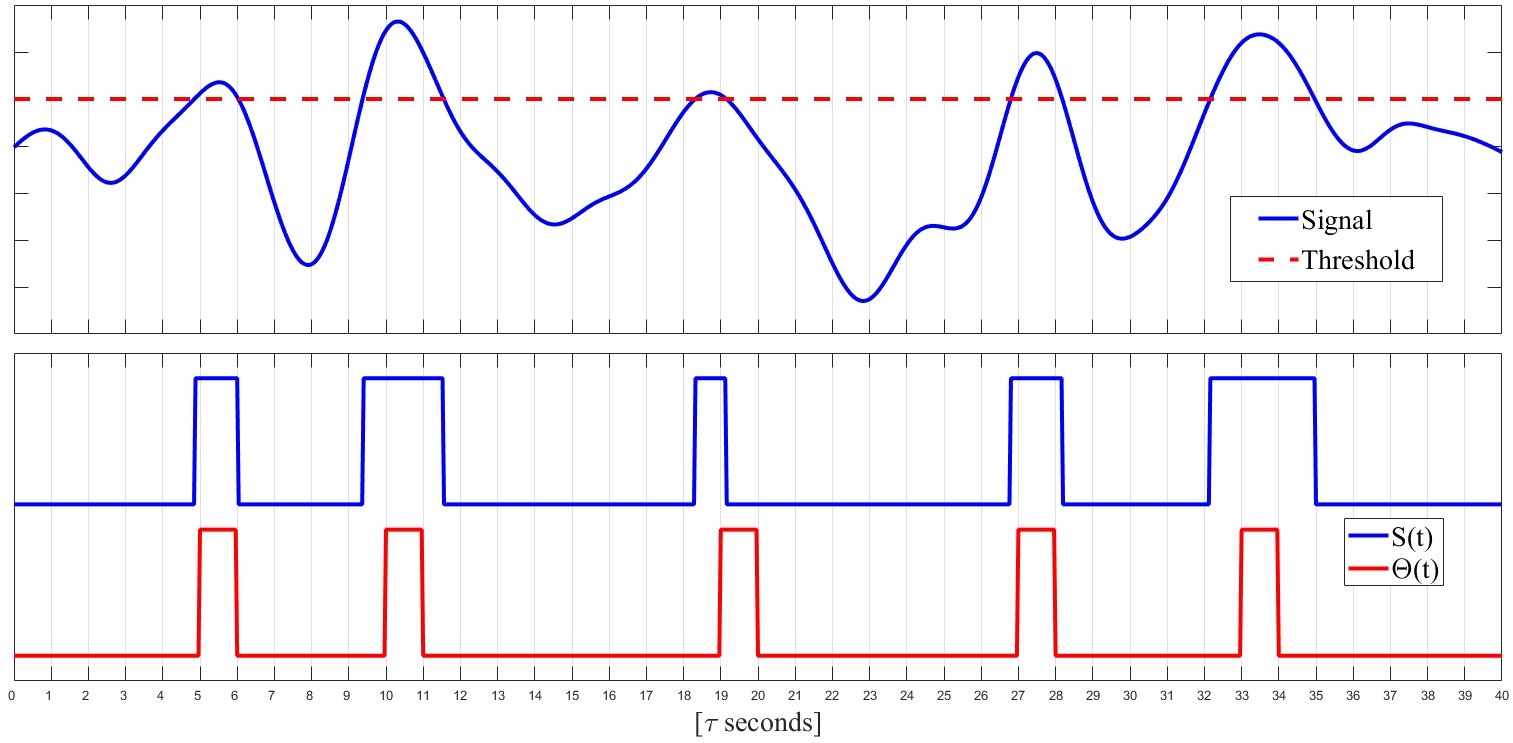}
    \caption{Example of a signal and its respective trigger function $S(t)$ and event function $\theta(t)$.}  
    \label{fig: exp_signal}
\end{figure}
Let us illustrate this by assuming that the physical signal refers to the temperature in a given room. 
The sensor monitors the signal and only acquires the data if the temperature rises above a given threshold.
When it does, this information needs to be transmitted to the sink node through a given communication channel, which is also available to the other users.
The sink node needs to indicate if the event, i.e., the threshold crossing, has happened in that specific node. 
This can be formalized as follows.
Let the \textit{trigger function} $S: \mathbb{R} \rightarrow \{0,1\}$ be defined as 
\begin{IEEEeqnarray}{lcl}
S_{\mathcal{E}}(t)=
\begin{cases}
0, \text{if the signal $\mathcal{E}$}<\text{threshold}\\
1, \text{otherwise}
\end{cases},
\end{IEEEeqnarray}
where $\mathcal{E}$ refers to the physical process monitored by one of the $N$ sensors.

Defining a time interval $\mathcal{T}_n=[(n-1) \tau,n\tau)$, where 
$n \in \mathbb{N}$
and $\tau \in \mathbb{R}^+$, we can then define the \textit{event function} $\theta_{\mathcal{E}}(t)$ with $\theta : \mathbb{R} \rightarrow \{0,1\}$ indicating whether an event occurred during $\mathcal{T}_n$ on the signal $\mathcal{E}$ as
\begin{IEEEeqnarray}{lcl}
    \theta_{\mathcal{E}}(t)=
    \begin{cases}
        1,& \text{if } \exists \, x, y\in \mathcal{T}_{\lfloor t/\tau \rfloor}, \, x < y \\ & \text{such that } S_{\mathcal{E}}(x) < S_{\mathcal{E}}(y) \\
        0,& \text{otherwise}
    \end{cases},
\end{IEEEeqnarray}
in which $\lfloor \cdot \rfloor$ rounds to the largest integer less than or equal to the argument.
At the sink, the event function $\theta_{\mathcal{E}}(t)$ is estimated by $\hat{\theta}_{\mathcal{E}}(t)$. 
Hereafter, we refer to an event (i.e., $\theta_{\mathcal{E}}(t)=1$) occurring on the signal $\mathcal{E}$ by event $\mathcal{E}$ in order to simplify the notation.

This communication strategy is suitable for rare events like alarms (e.g., the temperature has crossed the safety threshold).
In this case, the threshold defines the event that must have a meaning in relation to the physical process being processed. 
At the sink, the knowledge of the occurrence of the event communicated by a specific node must be functional in relation to its meaning; in this case, trigger an alarm for the specific room where the temperature is being monitored by that sensor.
Following the terminology of \acp{WSN}, the sink node acts as a fusion center, which is capable of concentrating the readings of several sensor nodes~\cite{ALSKAIF2017141}.

\textbf{Remark:} This scheme is constructed assuming that the information to be transmitted has a well-defined meaning (temperature and safety of the room) that will have a functional role (as part of the alarm system). For this reason, we call this \textit{semantic-functional communication}.

Note also that alarm monitoring is a primary issue in some environments like chemical manufacturing, oil refineries, petrochemical facilities, and power plants, and different communication strategies have been reported in the literature~\cite{Wang7270356,santos2020performance}. 
Our aim here is to build an effective way to communicate events with a low error probability, while efficiently using the communication network resources. 
Before going into the details of the proposed solution, we now point out the key differences with other medium access control (MAC) techniques that could perform this task.

\subsection{Comparison with Different MAC Protocols}
The scenario described above can be deployed using standard MAC protocols, either random access or channel partitioning.
The random type refers to protocols that attempt to minimize collisions through random decisions (e.g., ALOHA). 
Channel partitioning protocols centrally share the available resources (e.g., time or frequency), which can be done dynamically or statically  (e.g., variations of \ac{TDMA}). 
The class of random \ac{MAC} protocols have the advantage of having a low complexity of control over the nodes, but it either has a poor performance in terms of packet collision or should be complemented with error detection and retransmission schemes.
Static channel partitioning protocols split the shared resources to guarantee the reliability of the package delivery reliability while using a relatively small control traffic. 
Dynamic channel partitioning protocols require  relatively large control traffic compared with the message size considered in the proposed scenario (this last option will not be considered here).
In summary, we will focus on a comparison between the proposed \acs{SFC} system, slotted ALOHA, and \ac{TDMA}.

In \ac{TDMA}, all sensor nodes are synchronized, and thus, fixed and equal portions of time are allocated to each sensor node. 
In this way, a sensor node uses its predefined time slot to communicate an alarm to the sink node, which avoids collisions. 
In this hypothetical scenario, only a small amount of control traffic is used in order to keep the nodes synchronized.
Further, the message could be exempt from a header because each sensor can be identified by the time slot used.
In addition, the maximum latency would be the \ac{TDMA} frame length with assured delivery and reproduction of the event function $\theta_{\mathcal{E}}(t)$ on the sink node side.

Slotted ALOHA can also be easily applied to our scenario, considering that each sensor node has a unique identity to be transmitted to the sink node together with the message.
ALOHA does not require synchronization, although slotted ALOHA improves the communication success.

In the following section, we will describe the proposed solution, aiming to overcome the drawbacks of both channel partitioning and random protocols, represented by \ac{TDMA} and slotted ALOHA, respectively. 
As will become clear throughout this paper, the proposed physical layer encompasses aspects of modulation, encoding, and MAC in a single layer.
In this sense, the proposed design is not an improved version of any existing \ac{MAC} protocol or modulation technique. 
The idea we present next is rather to establish a truly \ac{SFC} system  based on event-triggered sampling for alarm messages.

\section{Proposed Semantic-Functional Communication System for Alarm Messages}
\label{sec:sys}
In this section, we will describe the proposed solution to transmit alarms, following Fig. \ref{fig: scenario}.
Specifically, the network topology and the proposed physical layer will be presented, together with the semantic way to code and decode the transmitted alarm message.

\subsection{Framed Structure of Network Resources}
{The proposed approach} assumes a framed structure of the network resources.
Let $\mathbb{S}$ be a set of $N$ code words associated with each sensor node illustrated in Fig. \ref{fig: scenario}, and thus, the most compact way to represent them is by $k= \lceil \log_2 N \rceil$ bits, where $\lceil \cdot \rceil$ rounds to the smallest integer larger than or equal to the argument. 
Within our context, each event $\mathcal{E}$ is defined by a unique set of $k$ bits, i.e., by a code word.
Let $R$ be the number of subsets of the shared resource. 
For example, if the shared resource is the frequency spectrum with a bandwidth $W$, we would have $R$ subcarriers of width $W/R$ Hz each.

Fig. \ref{fig: symb struc} illustrates the proposed structure for the radio frame.
We map each code word using only one energy slot per subsymbol; therefore, the receiver must evaluate $k$ subsymbols to make a decision (more details will be provided later in this section). 
In addition, the receiver must regard any ensemble of $k$ successive subsymbols as one radio symbol.
For instance, the next radio symbol with respect to Fig. \ref{fig: symb struc} is composed of subsymbols $\#2, \#3,\dotsc, \#k+1$.
A subsymbol, in turn, is built out of a set of $R$ energy slots.
This is the basic structure of our communication system, which, to the best of our knowledge, is novel and different from existing systems.

\begin{figure}[t!]
    \centering
    \includegraphics[width=1.1\linewidth]{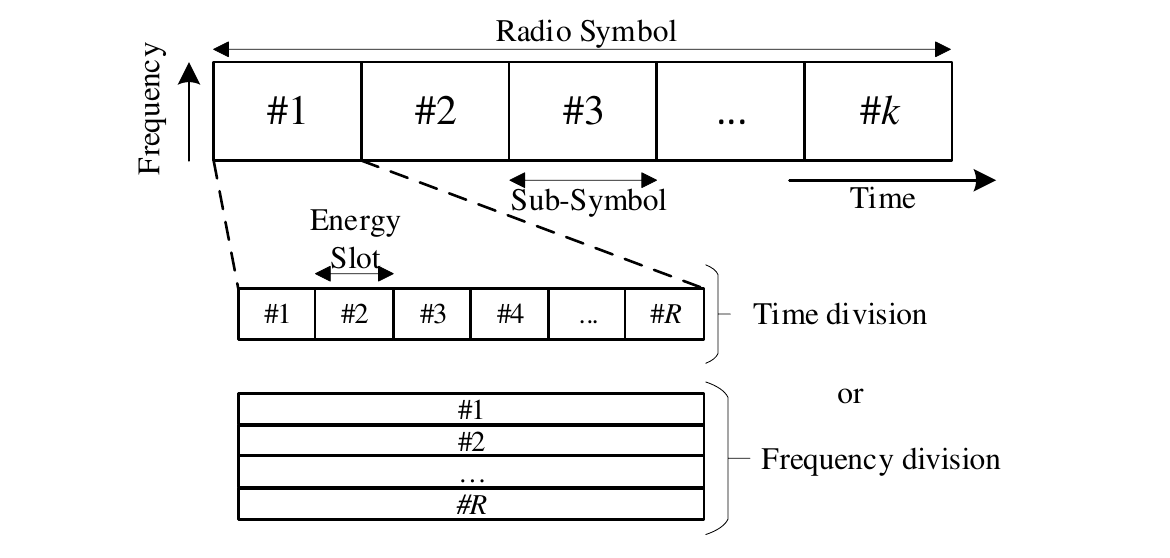}
    \caption{Radio symbol structure.}  
    \label{fig: symb struc}
\end{figure}

As mentioned above, the energy slot is the basic element of our proposed approach. 
The energy slot represents an interval at which the receiver must estimate the signal energy. 
The receiver does not demodulate the energy slot, because it does not actually carry information as a modulated signal. 
The goal of the receiver is to determine whether or not there is a signal in an energy slot by estimating its energy. 
Accordingly, it is not required to synchronize the phase of the local carrier at the receiver.
%

Although the lack of synchronization and allocation control can reduce the energy required to send a code word, because of the preamble and control messages not being used, it may nevertheless require that the receiver constantly observes the channel. 
Consequently, the receiver has the uninterrupted task of detecting a transmitted pulse or inferring that the medium was idle during the duration of an energy slot. 
On the other hand, the sensor nodes can instantly send the information when they detect an event. 
\subsection{Transmission of a Given Event}
Consider a network composed of $N$ sensor nodes that uninterruptedly monitor continuous signals, each of them associated with a unique code word.
The purpose of sensor nodes is to recognize whether a given predetermined event $\mathcal{E}$ has happened in order to feed this information to an alarm system centrally controlled at the sink. 
As discussed in Sec. \ref{sec:preli}, we consider a multiuser scenario where sensor nodes share the same communication resources, as defined in the previous subsection, to transmit whether the event associated with the corresponding measurements has happened.
%

Without loss of generality, we assume frames that are divided only in the time domain. 
Let us define a discrete version of the event function $\theta_{\mathcal{E}}(t)$ as $\theta_{\mathcal{E}}[n]=\theta_{\mathcal{E}}(n\tau)$ with $\theta : \mathbb{N} \rightarrow \{0,1\}$.
Note that the sink node can easily build $\theta_{\mathcal{E}}(t)$ from $\theta_{\mathcal{E}}[n]$.
By setting the parameter $\tau \in \mathbb{R}^+$ as the duration of a subsymbol, we can then create the following event transmission rule:
%
if $\theta_{\mathcal{E}}[n]=1$ (i.e., an alarm is detected), then the sensor node starts transmitting $k$ subsymbols in sequence starting at a discrete time $n$.
This is illustrated in Fig. \ref{fig: trans chain}, where the state diagram of the sensor node is presented.

In order to illustrate how the sensor node manages event transmission, let $\mathbf{Y}_n$ be a $k \times R$ matrix representing the radio symbol at a discrete time $n$ (i.e., the radio symbol whose first subsymbol is in $\mathcal{T}_n$) with the element $y_{i,j} \in \{0,1,\dotsc,N\}$ in the $i\text{th}$ row and the $j\text{th}$ column representing the energy inside the energy slot $j$ of the subsymbol $i+n$. 
Further, let $\mathbf{C}_{\mathcal{E}}$ be a $k \times R$ matrix with the element $c_{i,j} \in \{0,1\}$ in the $i\text{th}$ row and the $j\text{th}$ column. 
The matrix $\mathbf{C}_{\mathcal{E}}$ then represents the binary mapping of the event $\mathcal{E}$ into $k$ subsymbols (more details will be provided later in Section \ref{sec:mapping}).
To simplify the notation, we henceforth refer to the binary transmission mapping matrix of the event $\mathcal{E}$, i.e., $\mathbf{C}_{\mathcal{E}}$, by the transmission map $\mathcal{E}$.
For simplicity, we assume that only two events occur in a discrete time $n$. 
Furthermore, we assume that the transmission map for one of these events is given by the matrix $\mathbf{C}_{1}$ and for the other event given by the matrix $\mathbf{C}_{2}$.
Therefore, the $n\text{th}$ radio symbol is given by $\mathbf{Y}_{n} = \mathbf{C}_{1} + \mathbf{C}_{2}$.
In other words, the sensor node, which monitors the signal $\mathcal{E}$, adds energy (i.e., it transmits an unmodulated carrier for the duration of an energy slot) to the respective energy slots of the $n$th radio symbol in accordance with the values of $c_{i,j}$.
Note that the sensor does not need any knowledge of the channel status to transmit.
\begin{figure}[t!]
    \centering
    \includegraphics[width=0.7\linewidth]{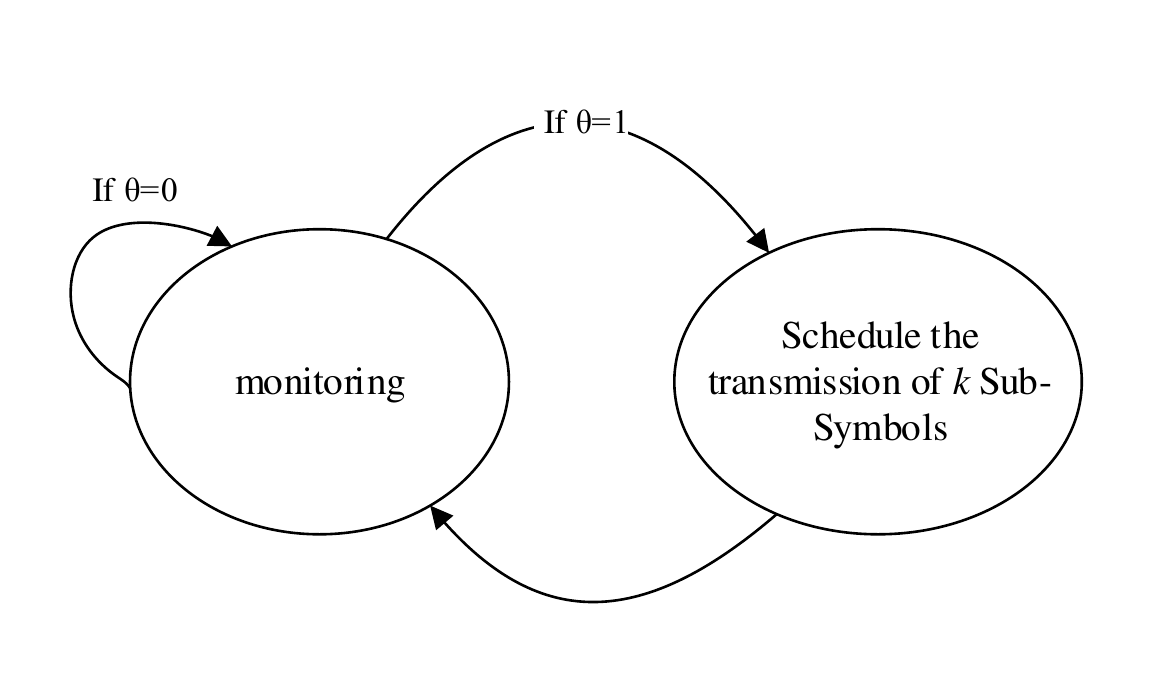}
    \caption{State chain of a sensor node.}  
    \label{fig: trans chain}
\end{figure}

\subsection{Reception of a Given Event}
On the reception side, the sink node needs to correctly identify the occurrence of a predetermined event related to the specific node that transmits it.
The recovery of this semantic information is straightforward in our approach, as it does not require any demodulation or data processing (recall that the transmitted message is neither modulated as a waveform nor composed of higher layer overheads). 
In this case, if the function of the sink is to operate an alarm system based on the occurrence of predetermined events, then all relevant information for that function can be recovered simply by receiving or not a code word. 
Note that even when no sensor transmits, the sink continues to acquire information about the physical processes, because it can correctly identify if the channel is idle; in other words, silence in the channel is informative in its own right, and the proposed semantic-functional approach makes use of this fact.

The sink node, in turn, must constantly monitor the channel and create a log of the received subsymbols, and then search for valid transmission maps.
The process of identifying an event is as follows. 
Let $\mathbf{H}_n$ be a $k \times R$ matrix containing the last $k$ subsymbols received (i.e., one radio symbol) at a discrete time $n \geq k$ with the element $h_{i,j} \in \{0,1,\dotsc,N\}$ in the $i\text{th}$ row and the $j\text{th}$ column representing the estimated energy slot (``1'' for used, ``0'' otherwise) $j$ within the received subsymbol $i-k+n$. 
%
%
In order to find out if the transmission map $\mathcal{E}$ was transmitted at a discrete time $n$, the sink must first perform a point-by-point multiplication between $\mathbf{C}_{\mathcal{E}}$ and $\mathbf{H}_n$ and then count how many non-null elements the resulting matrix has.
If the counting result is $k$, the sink decides in favor of the event $\mathcal{E}$, otherwise not.

The challenge is now to construct a method to generate transmission maps that can uniquely map the event related to a given sensor, which can then be decoded at the sink in the presence of multiple users.
In the next section, we will describe the proposed random map in relation to conventional mappings based on the MAC protocols we discussed above.
\section{Mapping Sensor-Defined Events into Transmission Maps}
\label{sec:mapping}

This section presents how to map the predetermined events corresponding to sensor observations into transmission maps defined by the sequence of energy slots that define subsymbols.
After introducing the proposed method based on a random map, we will indicate how the scheme defined in the previous section can be deployed using conventional modulation and media access control protocols.
\subsection{Random Map} \label{sec: Random}
Let us assume that the total number of possible transmission maps is given by $R^k$, of which $N$ transmission maps are used to map all possible events. 
Let $T$ and $E$ be the number of  transmitted transmission maps and the number of erroneously received transmission maps for a radio symbol, respectively.
Now, suppose a radio symbol contains $F$ transmission maps that do not map any event.
Hence, there are $F$ transmission maps that carry no semantic information, and thus, are ignored by the sink during the reception process because they do not belong to the set of valid code words that map events. 

Note that we have $F+T+E=Q$ transmission maps for each radio symbol. 
Thereby, assuming that the $R^k$ possible transmission maps are equiprobable, the probability that $E$ transmission maps are incorrectly received given $Q$ can be computed as 
\begin{IEEEeqnarray}{lcl}\label{eq: prob E}
    P[E|Q,T] = \frac{\binom{Q-T}{E} \binom{R^k-Q}{N-T-E}}{\binom{R^k-T}{N-T}}
\end{IEEEeqnarray}
for $T \leq N \leq R^k$ and $T+E\leq Q \leq R^k$. 

Equation \eqref{eq: prob E} can be understood as the probability of obtaining $E$ favorable outcomes in an experiment where $F+E$ numbers are chosen from a set of $R^k-T$ options, and a random draw without replacement and ordering of $N-T$ numbers is performed.
Note that in the proposed approach there are always exactly $T$ transmission maps (correctly) transmitted, and thus, the number $R^k$ of elements in the sampling space of incorrectly received transmission maps needs to be reduced to $R^k-T$. 
%
It is worth mentioning that \eqref{eq: prob E} refers to a hypergeometric distribution.

We can now present an interesting fact about the proposed solution.
\begin{lemma}
    The error probability ${P_\emph{e}(Q,T) \triangleq 1-P[E=0|Q,T]}$ can be lower and upper bounded by
    \begin{align*}
    1- \left(1-\frac{N-T}{R^k+1}\right)^Q \frac{\binom{R^k}{N-T}}{\binom{R^k-T}{N-T}} \le P_\emph{e}(Q,T)
    \end{align*}
    and
    \begin{align*}
    P_\emph{e}(Q,T) \le 1-\left(1-\frac{N-T}{R^k-Q+1}\right)^Q \frac{\binom{R^k}{N-T}}{\binom{R^k-T}{N-T}}.
    \end{align*}
\end{lemma}
\begin{proof}
    We start with the probability that the system has at least one error, which can be computed as $1-P[E=0|Q,T]$. 
Then, we can rewrite \eqref{eq: prob E} as $P[E=0|Q,T]=p_{\text{E}}(Q,T)=\alpha(Q) p_{\text{E}}(Q-1,T)$ $\forall Q \geq 1$ with $p_{\text{E}}(0,T) \triangleq \lim_{Q \rightarrow 0}P[E=0|Q,T]$ where, after some mathematical manipulation, we have
\begin{IEEEeqnarray}{lcl}
    \alpha(Q) = \frac{p_{\text{E}}(Q,T)}{p_{\text{E}}(Q-1,T)} = 1-\frac{N-T}{R^k-Q+1}.
\end{IEEEeqnarray}

Hence, the error probability as a  function of the variables $Q$ and $T$ is written as
\begin{IEEEeqnarray}{lcl}\label{eq: prob Error}
    P_{\text{e}}(Q,T) = 1-\frac{\binom{R^k}{N-T}}{\binom{R^k-T}{N-T}}\prod_{i=1}^{Q}\alpha(i).
\end{IEEEeqnarray}

We can postulate that 
\begin{align*}
    1-\frac{N-T}{R^k-Q+1} \leq \alpha(i) \leq 1-\frac{N-T}{R^k+1}    
\end{align*}
for all $i \in \{1,2\dots,Q\}$, because
\begin{align*}
    1-\frac{N-T}{-i+R^k+1} &\geq 1-\frac{N-T}{R^k-Q+1} \\
    \frac{N-T}{R^k+1-i} &\leq  \frac{N-T}{R^k+1-Q} \\
    0 &\leq (N-T)(Q-i)
\end{align*}
and
\begin{align*}
    1-\frac{N-T}{R^k+1-i} &\leq  1-\frac{N-T}{R^k+1} \\
    \frac{N-T}{R^k+1-i} & \geq \frac{N-T}{R^k+1} \\
    0 & \leq i (N-T).
\end{align*}
Therefore, we have
\begin{align*}
    \left( 1-\frac{N-T}{R^k-Q+1} \right)^Q \leq \prod_{i=1}^{Q} \alpha(i) \le  \left( 1-\frac{N-T}{R^k+1} \right)^Q
\end{align*}
    
\end{proof}

Thus far, we have only studied the overall error probability, which is evaluated in terms of the vector $\hat{\Theta}=(\hat{\theta}_1, \hat{\theta}_2,\dotsc,\hat{\theta}_N)$.
In this case, all $N$ sensors are considered, i.e., the error probability quantifies the overall reliability of the system. 
In other words, the error probability quantifies how accurate the estimation $\hat{\Theta}$ is.
%
Thereby, we can understand how prone to error our communication system is. 

On the other hand, it is important to evaluate how accurate the estimation of a specific event function related to one given sensor is.
In this case, we now focus our analysis on the error probability related to $\hat{\theta}_{\mathcal{E}}$ for an arbitrary sensor node.

At this point, we can state the error probability of a single event.
\begin{lemma}\label{lem: prob s event}
    The error probability for a single event function is given by 
    \begin{IEEEeqnarray}{lcl}\label{eq: error prob theta given F long}
        P_{\emph{e}}^{\hat{\theta}_{\mathcal{E}}}(Q,T) = \frac{Q-T}{R^k-T}P[\theta_{\mathcal{E}}=0].
    \end{IEEEeqnarray}
\end{lemma}
\begin{proof}
    The error probability related to $\hat{\theta}_{\mathcal{E}}$ for an arbitrary sensor node can be written as 
    \begin{IEEEeqnarray}{lcl}
        P_{\text{e}}^{\hat{\theta}_{\mathcal{E}}} &=& P[\hat{\theta}_{\mathcal{E}}=1|{\theta}_{\mathcal{E}}=0]P[{\theta}_{\mathcal{E}}=0] \nonumber \\
        &+& P[\hat{\theta}_{\mathcal{E}}=0|{\theta}_{\mathcal{E}}=1]P[{\theta}_{\mathcal{E}}=1].
    \end{IEEEeqnarray}

    We have $P[\hat{\theta}_{\mathcal{E}}=0|{\theta}_{\mathcal{E}}=1]=0$  for our \acl{SFC} system, and $P[\hat{\theta}_{\mathcal{E}}=1|{\theta}_{\mathcal{E}}=0]$ can be evaluated by \eqref{eq: prob E} with $N=T+1$.
    In the case of a single event function, \eqref{eq: prob E} can be interpreted by $E\in\{0,1\}$ favorable outcomes in an experiment where $F+E$ numbers are chosen from a set of $R^k-T$ options, and a random draw of only one number is performed.
    %
    %
    %
    Substituting $N=T+1$ into \eqref{eq: prob E} and, after some mathematical manipulation, we find the error probability for a single event function given $Q$ and $T$ as
    \begin{IEEEeqnarray}{lcl}
        P[E=1|Q,T] = P_{\text{e}}^{\hat{\theta}_{\mathcal{E}}}(Q,T)  &=\frac{Q-T}{R^k-T}P[\theta_{\mathcal{E}}=0] 
    \end{IEEEeqnarray}
    %
%
\end{proof}

\begin{lemma}\label{lem: bound s event}
     The error probability for a single event function $P_{\emph{e}}^{\hat{\theta}_{\mathcal{E}}}(Q,T)$ can be lower and upper bounded by 
     \begin{IEEEeqnarray}{lcl}
         \frac{Q-T}{R^k}P[\theta_{\mathcal{E}}=0] \leq P_{\emph{e}}^{\hat{\theta}_{\mathcal{E}}}(Q,T) \leq \frac{Q-T}{R^k-N}P[\theta_{\mathcal{E}}=0].
     \end{IEEEeqnarray}
\end{lemma}
\begin{proof}
    Since $0 \leq T \leq N$, the limits are taken directly from \cref{lem: prob s event}.
\end{proof}

\begin{corollary}
    Since $R^k \to \infty$, the bounds in \cref{lem: bound s event} converge to the exact value of the error probability for a single event function $P_{\emph{e}}^{\hat{\theta}_{\mathcal{E}}}(Q,T)$.
\end{corollary}
\begin{proof}
    Let 
    \begin{IEEEeqnarray}{lcl}
        \delta = \frac{Q-T}{R^k-N}P[\theta_{\mathcal{E}}=0] - \frac{Q-T}{R^k}P[\theta_{\mathcal{E}}=0]=\frac{N(Q-T)}{R^k (R^k-N)}P[\theta_{\mathcal{E}}=0]
     \end{IEEEeqnarray}
    be the difference between the bounds in \cref{lem: bound s event}, we have that
    \begin{IEEEeqnarray}{lcl}
        \lim_{R^k \to \infty} \delta = 0.
     \end{IEEEeqnarray}
\end{proof}
Note that, at this point, no explicit mapping is considered.
One mapping option is to completely avoid transmission map combinations, somehow building orthogonal maps.
For example, suppose you want to add a third sensor node to the network. 
In this case, in order to build orthogonal transmission maps, there would be $R^k-2^k$ transmission maps available for the third sensor node ($2^k$ is the minimum number of possible combinations between two transmission maps) at best.
If a fourth sensor node was added, there would be $R^k-3^k$ transmission maps available for the fourth sensor node at best. 
Note how drastic this reduction is as new sensor nodes are added to the network.
Furthermore, it is clear that we would have a limit of $R$ node sensors operating with orthogonal transmission maps.
%
%
Therefore, orthogonal transmission maps would either imply an extremely large $R$ or a very low number of sensors $N$.
The second option is to randomly select transmission maps, in which case transmission maps may overlap, and thus, all $R^k$ are feasible; if $R^k \gg T$, then the likelihood of transmission map combinations that lead to errors is very low.
The third option is to build maps that try to avoid such overlaps but do not entirely forbid them.
%
In this way, we would have a certain degree of freedom (or, equivalently, a prohibition degree) over the mapping.
Going back to the previous example, the addition  of the third sensor node in the network would lead to a reduction in the maximum number of sensor nodes smaller than $3^k$, as small as the adopted degree of freedom.
In any case, if we consider $R^k$ fixed, any restrictive solution that reduces such sample space from where the transmission maps are drawn  would lead to a poorer performance.
%
Several rules can be adopted for the creation of transmission maps, but for simplicity, we adopt a random mapping.

The proposed random transmission map selection is presented in the following definition.
\begin{definition}[Random map]
Let $\mathbf{C}_{\mathcal{E}}$ be the binary matrix  with the element $c_{i,j} \in \{0,1\}$ that denotes a transmission ($c_{i,j}=1 $) or not ($c_{i,j}=0 $) during the $j\text{th}$ energy slot of the $i\text{th}$ subsymbol, to be built as follows.
For each row $i$, only one element $c_{i,j}$ is randomly set to 1. 
The choice of the element $c_{i,j}$ set to $1$ follows a uniform distribution.
If one or more code words have the same matrix $\mathbf{C}_{\mathcal{E}}$, one must repeat the mapping process for them until there is no repeated matrix $\mathbf{C}_{\mathcal{E}}$.
\end{definition}

\subsection{Conventional Mapping} \label{sec: Conventional}
%

The main difference from the proposed mapping based on energy slots is that the proposed system has only two states: ``0'' and ``1''.
While modulations like \ac{BPSK} or even \ac{OOK} use a similar binary map, the receiver has an additional, preliminary state to detect: the channel idleness. 
In the case of traditional  digital modulation, the sink needs to detect (i) if the channel is in the state ``idle,'' and, if not, (ii) what bit has been transmitted.
The proposed random map assumes that the channel state ``idle'' and ``0'' are the same, building on an ``informative silence.''


Regarding the \ac{MAC} protocol for our comparison benchmark, we opted for the simplest channel partitioning that provides the same channel capacity, namely \ac{TDMA}~\cite{BookTorrieri}.
Note that \ac{FDMA} and \ac{CDMA} could also be employed, but we prefer \ac{TDMA} because of its simplicity when presenting numerical results.
We can implement it directly into our physical layer as follows. 
Let $N_\text{t}$ be the number of sensor nodes. 
Since $R=N_\text{t}$ and suitably adjusting the transmission powers, one obtains a classic \ac{TDMA} (or \ac{FDMA}) system. 
%
This channel partitioning \ac{MAC} protocol requires an unnecessarily large resource volume for an alarm system scenario that uses event-triggered data transmission, as discussed in the previous sections.

In this case, one could expect that the occurrence of an alarm is uncommon, and thus, the sensor nodes will rarely occupy the resources allocated to them. 
Thereby, channel partitioning \ac{MAC} protocols would lead to a system with high channel idle rates and overallocation of resources. 
One could argue in favor of channel partitioning MAC protocols with dynamic resource allocation.
In general terms, the available resources are allocated on-demand to the sensor nodes for these MAC protocols. 
Therefore, the sensor nodes must somehow request the use of the resource. 
The control load generated only in the resource request stage would be as costly as the direct sending of information about the event itself, because when requesting the allocation of transmission resources, the sensor node would indirectly be informing the sink about the event.
In addition, the increase in latency is evident when adopting dynamic channel partitioning MAC protocols.

For comparison purposes, let us suppose that the number of available resources is insufficient for a network containing $N_\text{t}$ sensor nodes to employ collision-free \ac{TDMA}, i.e., we have $R<N_\text{t}$.
The division of $R$ available resources for this case is carried out as follows. 
Let $\mathbb{S}$ be a set of $N_\text{t}$ code words and $\mathbb{S}_i$, $i \in \{1, 2,\dotsc, R\}$, be a subset of $\mathbb{S}$ such that $\mathbb{S}_i \cap \mathbb{S}_j = \emptyset$ $\forall j\neq i$. 
We assume also that the number of elements in each subset $\mathbb{S}_i$ is as close to $N_\text{t}/R$ as possible.
In order to send a code word from the set $\mathbb{S}_i$, the sensor node must transmit a \ac{BPSK} symbol during the $i\text{th}$ energy slot of each subsymbol.
Note that a complete transmission of a code word requires $k$ \ac{BPSK} symbols and thus, $k$ subsymbols.

In addition to \ac{TDMA}, we also adopted slotted ALOHA for comparison.
In this case, we assume that the slot has a duration of $k \tau /R$ seconds, and the packets have a fixed size of $k$ \ac{BPSK} symbols, thereby maintaining the same power consumption and bandwidth admitted for \ac{TDMA} and \ac{SFC}.
%
\section{Performance Analysis}
\label{sec: Perform}
%

Consider the scenario proposed in Fig. \ref{fig: scenario} assuming that the $N$ physical processes are statistically independent and that each sensor node monitors exclusively one signal.
We also assume that the time interval between the occurrence of the same alarm is an exponential variate of mean $N \tau/ \lambda$. 
Therefore, the total of observed events across all sensor nodes within an interval $\mathcal{T}_n$ follow a Poisson distribution of mean $\lambda$.
Let $C_n=\sum_{\mathcal{E}} \theta_{\mathcal{E}}[n]$ be the sum of the number of alarm signals within the interval $\mathcal{T}_n$.
Because the multiple occurrences of the same event within an interval $\mathcal{T}_{n-1}$ are signaled by $\theta_{\mathcal{E}}[n]=1$, $C_n$ does not follow a Poisson distribution. 
However, because $N/ \lambda \gg 1$, the probability that the time between the occurrence of two consecutive alarms for the same signal is smaller than $\tau$ is negligible.
Under this condition, we can say that $C_n \sim \text{Poi}(\lambda)$, as we assume hereafter.

The total number of accesses to the $n\text{th}$ subsymbol is given by $A_n=\sum_{i=n-k+1}^{n}C_i$ $\forall n \geq  k$.
It is easy to show that $A_n$ is a Poisson variate of mean $k \lambda$ because the random variables $C_n$ are independent and identically distributed. 
Note that $A_{n}$ and $A_{n \pm i}$, $i\in \{1,2,\dotsc,k-1\}$, are dependent random variables, and that $T=C_n$ for the $n$th radio symbol.
Let us define the vector $\textbf{A}_n=(A_n,A_{n+1},\dotsc,A_{n+k-1})$ containing the number of accesses to each subsymbol of the $n\text{th}$ radio symbol.
To evaluate the probability $P[\textbf{A}_n]$, we must first calculate the probability of $\textbf{A}_n$ given $C^*_n=(C_{n}, \dotsc,C_{n-k+2})$.
%
This conditional probability can be calculated as
\begin{IEEEeqnarray}{lCl}\label{eq: P[An,An+1|C]}
	P[\textbf{A}_n|C^*_n] &= p_{\text{p}}(A_n-\mathcal{C}_n,\lambda) p_{\text{p}}(A_{n+1}-\mathcal{C}_n,\lambda) 
	\prod_{i=n+2}^{n+k-1} p_{\text{p}}(A_i-A_{i-1}+C_{i-k},\lambda), 
\end{IEEEeqnarray}
where
\begin{IEEEeqnarray}{lcl}
    p_{\text{p}}(a,b) = \frac{b ^{a} \exp (-b )}{\Gamma(1+a)}
\end{IEEEeqnarray}
is a Poisson \ac{pmf}, $\mathcal{C}_n=\sum_{i=n-k+2}^{n} C_i$, and $\Gamma(\cdot)$ is the Gamma function~\cite{bookGrandshteyni}.
In order to maintain the mathematical rigor and, at the same time, have the freedom to adopt any values for $A_n$, $A_{n+1}$, and $C_{n-k+1}$, we opted to define $p_{\text{p}}(\cdot,\cdot)$ using the Gamma function in the denominator instead of the factorial.
By taking an average over all possible values of $C^*_n$, the probability of $\textbf{A}_n$ is obtained as
\begin{IEEEeqnarray}{lcl}
    P[\textbf{A}_n] = \sum_{C^*_n}\text{P}[\textbf{A}_n|C^*_n] P[C^*_n],
\end{IEEEeqnarray}
where 
\begin{IEEEeqnarray}{lcl}
    P[C^*_n] = \prod_{i=n-k+2}^{n}p_{\text{p}}(C_i,\lambda).
\end{IEEEeqnarray}


We will now deal with the probability of getting $U_n$ energy slots filled at the $n$th subsymbol given $A_n$. 
This probability can be assessed using the counting method.
We have $R^{A_n}$ equiprobable outcomes for $A_n$ sensor nodes sharing a subsymbol. 
The number of ways to partition a set of $A_n$ objects into $U_n$ nonempty subsets is, by definition, the Stirling number of the second kind, denoted by $\stirling{A_n}{U_n}$.
There are also $R!/(R-U_n)!$ arrangements of the filled energy slots. 
Hence, the probability of having $U_n$ given $A_n$ filled energy slots is given by
\begin{IEEEeqnarray}{lcl}
    P[U_n|A_n] = \frac{R!}{R^{A_n}(R-U_n)!} \stirling{A_n}{U_n}.
\end{IEEEeqnarray}

Let us define the vector $\textbf{U}_n=(U_n,U_{n+1},\dotsc,U_{n+k-1})$ containing the number of filled energy slots for each subsymbol of the $n\text{th}$ radio symbol. 
The probability of $\textbf{U}_n$ given $\textbf{A}_n$ can be directly calculated by
\begin{IEEEeqnarray}{lcl}\label{eq: prob Un given An}
    P[\textbf{U}_n|\textbf{A}_n] = \prod_{i=0}^{k-1} P[U_{n+i}|A_{n+i}].
\end{IEEEeqnarray}
Besides, the probability of $\textbf{U}_n$ is given by%
\begin{IEEEeqnarray}{lcl}
    P[\textbf{U}_n] = \sum_{\textbf{A}_n} P[\textbf{U}_n|\textbf{A}_n] P[\textbf{A}_n].
\end{IEEEeqnarray}

We are now ready to present the transmission error probabilities for the three different transmission approaches.
\begin{proposition}[Transmission error probability]
    The transmission error probabilities of \ac{SFC}, slotted ALOHA (sALOHA), and TDMA are given by
    \begin{IEEEeqnarray}{lcl}
        \label{eq: prob error final}
        \hspace{-3ex} P_{\emph{SFC}}^{\hat{\Theta}} = \sum_{T,\textbf{U}_n} P_{\emph{e}}\left( \prod_{i=0}^{k-1}U_{n+i}, T\right) P[\textbf{U}_n]p_\text{p}(T,\lambda),
    \end{IEEEeqnarray}
    \begin{IEEEeqnarray}{lcl}
        \label{eq: prob error sALOHA}
        P_{\emph{sALOHA}}^{\hat{\Theta}} =1-\frac{(k \lambda +R)}{R} \exp \left(-\frac{k \lambda }{R} \right),
    \end{IEEEeqnarray}
    and
    \begin{IEEEeqnarray}{lcl}
    \label{eq: prob error TDMA}
        P_{\emph{TDMA}}^{\hat{\Theta}} = 1-\left[\frac{\lambda}{R} \exp \left( -\frac{2 k \lambda }{R} \right)+\exp \left( -\frac{\lambda }{R} \right) \right]^R.
    \end{IEEEeqnarray}
\end{proposition}
\begin{proof}
    The \ac{SFC} transmission error probability is trivially achieved by using \eqref{eq: prob Error} and noting that $Q = \prod_{i=0}^{k-1}U_{n+i}$.
    
    The average error probability for slotted ALOHA, in our context, is directly given by the collision probability. 
    In other words, it is the probability of having two or more sensor nodes transmitting in a slot.
    Accordingly, in mathematical terms we have
    \begin{IEEEeqnarray}{lcl}
        P_{\text{sALOHA}}^{\hat{\Theta}} &=& 1-p_\text{p}(0,k \lambda/R) - p_\text{p}(1,k \lambda/R) \nonumber \\
        &=&1-\frac{(k \lambda +R)}{R} \exp \left(-\frac{k \lambda }{R} \right).
    \end{IEEEeqnarray}

    Regarding the \ac{TDMA} system defined in Sec. \ref{sec: Conventional}, we consider here only an arbitrary subset $\mathbb{S}_i$ of code words.
    Further, let random variables $E$ and $T$ represent the number of erroneously estimated code words and the number of transmitted code words, respectively, for that arbitrary subset $\mathbb{S}_i$.
    In this way, we can write
    \begin{IEEEeqnarray}{lcl}
    \label{eq: prob error S_i}
        P[E>0]=\sum_T P[E>0|T]P[T].
    \end{IEEEeqnarray}
    Furthermore, some conclusions can be drawn: (i) there is no error if no sensor node transmits; (ii) given that a sensor node has transmitted, communication is only error-free if no further transmission is initiated within the period of $(2k-1) \tau$ seconds (i.e., no collision); and (iii) given that two or more sensor nodes transmitted simultaneously, we will always have errors.
    Transcribing these conclusions into mathematical terms, we have
    \begin{align} \label{eq: E|T}
        \text{(i) }P[E>0|T=0]&=0, \nonumber\\
        \text{(ii) }P[E>0|T=1]&=1-p_\text{p}(0,(2k-1)\lambda/R), \\
        \text{(iii) }P[E>0|T>1] &= 1. \nonumber
    \end{align}

    The average error probability for TDMA is given by
    \begin{IEEEeqnarray}{lcl}
    \label{eq: E TDMA}
        P_{\text{TDMA}}^{\hat{\Theta}} = 1-(1-P[E>0])^R.
    \end{IEEEeqnarray}
    Plugging \eqref{eq: prob error S_i} and \eqref{eq: E|T} into \eqref{eq: E TDMA} and after some mathematical manipulation, we have
    \begin{IEEEeqnarray}{lcl}
        P_{\text{TDMA}}^{\hat{\Theta}} = 1-\left[\frac{\lambda}{R} \exp \left( -\frac{2 k \lambda }{R} \right)+\exp \left( -\frac{\lambda }{R} \right) \right]^R.
    \end{IEEEeqnarray}
\end{proof}

\begin{proposition}[Average error probability for a single event]
    \label{pre: error s event}
    The average error probability for a single event function can be lower and upper bounded by
    \begin{IEEEeqnarray}{lcl}
    \label{eq: error prob theta}
        \hspace{-3ex} \sum_{i=1}^{2^k} (-1)^{\alpha_i} \exp \left(  \frac{\beta_i \lambda}{(-R)^{\alpha_i}} \right)- \frac{\lambda}{R^k} \leq P_{\emph{e}}^{\hat{\theta}_{\mathcal{E}}} \exp \left(\frac{\lambda}{N} \right) \leq  \frac{R^k\sum_{i=1}^{2^k} (-1)^{\alpha_i} \exp \left(  \frac{\beta_i \lambda}{(-R)^{\alpha_i}} \right)- \lambda}{R^k-N}.
    \end{IEEEeqnarray}
\end{proposition}
\begin{proof}
    Using \cref{lem: bound s event}, we have
    \begin{IEEEeqnarray}{lcl}
    \label{eq: error prob theta mean F}
        \left(\frac{\mathbb{E}[Q]+T}{R^k-N} \right) \exp \left(-\frac{\lambda}{N} \right) \leq \mathbb{E}[P_{\text{e}}^{\hat{\theta}_{\mathcal{E}}}(Q,T)] \leq \left(\frac{\mathbb{E}[Q]+T}{R^k-N} \right) \exp \left(-\frac{\lambda}{N} \right),
    \end{IEEEeqnarray}
    with $\mathbb{E}[\cdot]$ denoting the expectation operator.
    
    Since $Q=\prod_{i=0}^{k-1}U_{n+i}$ and using \eqref{eq: prob Un given An}, we can write the expected value of $Q$ given $\textbf{A}_n$ as
    \begin{IEEEeqnarray}{lcl}
    \label{eq: expcted F}
        \mathbb{E}[Q|\textbf{A}_n]=\prod_{i=0}^{k-1}\mathbb{E}[U_{n+i}|A_{n+i}].
    \end{IEEEeqnarray}
    
    The expected value of $U_{n}$ given $A_{n}$ is given by
    \begin{IEEEeqnarray}{lcl}
    \label{eq: expcted Un given An}
        \mathbb{E}[U_n|A_n]=R\left(1-\left(1-\frac{1}{R}\right)^{A_{n}}\right).
    \end{IEEEeqnarray}
    
    Then, we have 
    \begin{IEEEeqnarray}{lcl}
    \label{eq: expcted F 0}
        \mathbb{E}[Q]&=\sum_{\textbf{A}_n} \prod_{i=0}^{k-1} R\left(1-\left(1-\frac{1}{R}\right)^{A_{i}}\right)P[\textbf{A}_n] \nonumber \\
        &= R^k \sum_{\textbf{A}_n} \sum_{i=1}^{2^k} (-1)^{\alpha_i} \left(1-\frac{1}{R}\right)^{\mathcal{q}_{i} \cdot \textbf{w}} P[\textbf{A}_n],
    \end{IEEEeqnarray}
    with the dot product $\mathcal{q}_{i} \cdot \textbf{w}$ representing all the $2^k$ exponents coming from the product of the terms $\left(1-\left(1-\frac{1}{R}\right)^{A_{i}}\right)$, $\alpha_i = \max \mathcal{q}_{i}$ is the largest value among the elements of $\mathcal{q}_{i}$, and $\textbf{w}=(C_{n-k+1},C_{n-k+2},\dotsc,C_{n+k-1})^\intercal$, where ``$^\intercal$'' denotes the transpose of a matrix or a vector.
     
    The possible values of $\mathcal{q}_{i}$ could be obtained as follows.
    Suppose there are $k$ sets, each containing two vectors of the length $2k-1$. 
    In all $k$ sets there is a null vector, i.e., $\textbf{0}=(0_1, 0_2, \dotsc, 0_{2k-1})$.
    The other vector of the set $l\in\{1,2,\dotsc,k\}$ has the first $l-1$ elements as zero followed by $k$ ones, and the last $k-l$ elements are zeros.
    $\mathcal{q}_{i}$ is obtained by any linear combination that takes a single vector from each of the $k$ sets described above.
    As an example, we present below a matrix whose rows represent the $2^k$ possible $\mathcal{q}_{i}$ vectors for $k=3$.
    \begin{align*}
        \begin{pmatrix}
            0  &  0 &    0&     0&     0\\
            1  &   1   &  1  &   0  &   0\\
            0  &  1  &   1  &   1  &   0\\
            1  &  2  &   2  &   1  &   0\\
            0  &   0  &   1  &   1  &   1\\
            1  &   1  &   2  &   1  &   1\\
            0  &   1  &   2  &   2  &   1\\
            1  &   2 &    3  &   2  &  1\\
        \end{pmatrix}
    \end{align*}
    In this context, e.g., $\mathcal{q}_{4}=(1, 2, 2, 1, 0)$ for $i=4$.
    
    Taking only the dependent terms of $\textbf{A}_n$ in \eqref{eq: expcted F 0}, we can write
    \begin{IEEEeqnarray}{lcl}
    \label{eq: sum An 0}
        \sum_{\textbf{A}_n} \left(1-\frac{1}{R}\right)^{\mathcal{q}_{i} \cdot \textbf{w}} P[\textbf{A}_n] = \exp \left(  \frac{ \beta_i \lambda}{(-R)^{\alpha_i}}  \right),
    \end{IEEEeqnarray}
    where
    \begin{align*}
        P[\textbf{A}_n]=\prod_{j=n-k+1}^{n+k-1}p_{\text{p}}(C_j,\lambda)
    \end{align*}
    and
    \begin{align*}
        \beta_i = \sum_{l=1}^{\alpha_i} (-R)^{\alpha_i-l} \sum_{j=1}^{2k-1} \frac{\mathcal{q}_{i,j}!}{(\mathcal{q}_{i,j}-l)!l!},
    \end{align*}
    with $\mathcal{q}_{i,j}$ representing the $j$th element of $\mathcal{q}_{i}$.

    Finally, by substituting \eqref{eq: expcted F 0} and \eqref{eq: sum An 0} into \eqref{eq: error prob theta mean F}, and after performing an average over all values of $T$, we can write the average error probability for a single event function as
    \begin{IEEEeqnarray}{lcl}
        \hspace{-3ex} \sum_{i=1}^{2^k} (-1)^{\alpha_i} \exp \left(  \frac{\beta_i \lambda}{(-R)^{\alpha_i}} \right)- \frac{\lambda}{R^k} \leq P_{\text{e}}^{\hat{\theta}_{\mathcal{E}}} \exp \left(\frac{\lambda}{N} \right) \leq  \frac{R^k\sum_{i=1}^{2^k} (-1)^{\alpha_i} \exp \left(  \frac{\beta_i \lambda}{(-R)^{\alpha_i}} \right)- \lambda}{R^k-N}.
    \end{IEEEeqnarray}
\end{proof}

We will now define the metric used to assess the performance of our physical layer. 
We have that $P[\hat{\theta}_{\mathcal{E}}=1|{\theta}_{\mathcal{E}}=1]$ and $P[\hat{\theta}_{\mathcal{E}}=0|{\theta}_{\mathcal{E}}=0]$ represent the reliability provided by the physical layer regarding the occurrence and nonoccurrence of an event, respectively.
We will therefore define an efficiency measure in order to quantitatively evaluate the performance of our physical layer to deliver the  desired information about the monitored signals to the sink as follows.
%
%

\begin{definition}[Efficiency metric]\label{def: Efficiency}
We can define the efficiency metric as a function of the probabilities $P[\hat{\theta}_{\mathcal{E}}=1|{\theta}_{\mathcal{E}}=1]$ and $P[\hat{\theta}_{\mathcal{E}}=0|{\theta}_{\mathcal{E}}=0]$  as
\begin{align}
    \mathcal{F} \triangleq P[\hat{\theta}_{\mathcal{E}}=1|{\theta}_{\mathcal{E}}=1] \; P[\hat{\theta}_{\mathcal{E}}=0|{\theta}_{\mathcal{E}}=0].
\end{align}
\end{definition}
%
%
%
The rationale for using such a metric is explained next.
Suppose that a given event is rare and that the physical layer always interprets that nothing was transmitted. 
In this case, we would have a high average of correctly received transmission maps.
However, the sink would never be able to detect alarm events.
On the other hand, by using the proposed metric, we would have $P[\hat{\theta}_{\mathcal{E}}=1|{\theta}_{\mathcal{E}}=1]=0$, and thus, $\mathcal{F}=0$; therefore, the efficiency $\mathcal{F}$ would capture this ineffectiveness of the physical layer in transmitting the desired event.
%

%
%

We can now calculate the efficiencies of the three studied strategies as presented next.
\begin{corollary}
    The efficiency $\mathcal{F}$ of \ac{SFC}, \ac{TDMA}, and slotted ALOHA are given by
    \begin{IEEEeqnarray}{lcl}
    \label{eq: F for SFC}
        \hspace{-3ex} \sum_{i=1}^{2^k} (-1)^{\alpha_i} \exp \left(  \frac{\beta_i \lambda}{(-R)^{\alpha_i}} \right)- \frac{\lambda}{R^k} \leq 1-\mathcal{F}_{\emph{SFC}} \leq  \frac{R^k\sum_{i=1}^{2^k} (-1)^{\alpha_i} \exp \left(  \frac{\beta_i \lambda}{(-R)^{\alpha_i}} \right)- \lambda}{R^k-N},
    \end{IEEEeqnarray}
    \begin{IEEEeqnarray}{lcl}
    \label{eq: F for TDMA}
        \mathcal{F}_{\emph{TDMA}}=\exp\!\left(-\frac{(2 k-1) \lambda }{R}\right)
    \end{IEEEeqnarray}
    and
    \begin{IEEEeqnarray}{lcl}
    \label{eq: F for s ALOHA}
        \mathcal{F}_{\emph{sALOHA}}=\exp\!\left(-\frac{(N-1) k \lambda }{N R}\right),
    \end{IEEEeqnarray}
    respectively.
\end{corollary}

\begin{proof}
    For \ac{SFC}, $P[\hat{\theta}_{\mathcal{E}}=1|{\theta}_{\mathcal{E}}=1]=1$ and $P[\hat{\theta}_{\mathcal{E}}=0|{\theta}_{\mathcal{E}}=0]=1-P_{\text{e}}^{\hat{\theta}_{\mathcal{E}}}\; \exp(\lambda/N)$. 
    In that way, the efficiency $\mathcal{F}$ of \ac{SFC} is obtained by using \cref{pre: error s event} into \cref{def: Efficiency}.
    
    Regarding \ac{TDMA} and slotted ALOHA, 
    we assume that in the event of a collision, the packets will be dropped by the receiver. 
    More precisely, we are assuming that the probability that the receiver recognizes a valid code word after a collision is negligible.
    Note that in this way, no code word is actually retrieved by the receiver in the event of a collision.
    Thereby, we have now that $P[\hat{\theta}_{\mathcal{E}}=0|{\theta}_{\mathcal{E}}=0]=1$ and $P[\hat{\theta}_{\mathcal{E}}=1|{\theta}_{\mathcal{E}}=1]=1-P[\hat{\theta}_{\mathcal{E}}=0|{\theta}_{\mathcal{E}}=1]$.
    Note that $P[\hat{\theta}_{\mathcal{E}}=0|{\theta}_{\mathcal{E}}=1]$ represents, in our context, for both TDMA and slotted ALOHA, the collision probability given that an arbitrary sensor node has transmitted.
    Hence, the efficiency $\mathcal{F}$ of TDMA and slotted ALOHA are given by
    \begin{IEEEeqnarray}{lcl}
        \mathcal{F}_{\text{TDMA}}&=&1-P[\hat{\theta}_{\mathcal{E}}=0|{\theta}_{\mathcal{E}}=1] \nonumber \\
        &=& 1-(1-p_{\text{p}}(0,(2k-1)\lambda / R)) \nonumber \\
        &=&\exp\!\left(-\frac{(2 k-1) \lambda }{R}\right)
    \end{IEEEeqnarray}
    and
    \begin{IEEEeqnarray}{lcl}
        \mathcal{F}_{\text{sALOHA}}&=&1-P[\hat{\theta}_{\mathcal{E}}=0|{\theta}_{\mathcal{E}}=1] \nonumber \\
        &=& 1-(1-p_{\text{p}}( 0,(N-1) k \lambda/(N R))) \nonumber \\
        &=&\exp\!\left(-\frac{(N-1) k \lambda }{N R}\right),
    \end{IEEEeqnarray}
    respectively.
\end{proof}


\section{Numerical Results} \label{sec: Numerical}
In this section, we evaluate the performance of the MAC protocols described above. 
%
%
In all figures illustrating the results, the markers represent the results obtained through simulation, and the lines are the analytical results obtained by the equations introduced above, i.e, \eqref{eq: F for SFC}\footnote{Because the bounds at \eqref{eq: F for SFC} are extremely tight, the results for the upper limit were omitted.}, \eqref{eq: F for TDMA} and \eqref{eq: F for s ALOHA}.
Numerical results were obtained through Monte Carlo simulations with the aid of Matlab software. 
TDMA and slotted ALOHA implement the modulation process described in Sec. \ref{sec: Conventional}. 
The \ac{SFC} follows the proposed model, described in Sec. \ref{sec: Random}.
We assume an ideal communication channel and postpone further analysis, such as the AWGN effect, fading, and phase error, for future work.
Thereby, the performance impact caused by collisions during transmissions will be more evident.

In Fig. \ref{fig: Eff x R}, the efficiency $\mathcal{F}$ for TDMA, slotted ALOHA and \ac{SFC} is plotted against the resource number $R$.
The $\lambda$ parameter varies as $\{ 0.1, 0.32 \}$, the code word length is $k=6$, and the number of sensor nodes is $N=64$. 
For all cases, as the number of resources allocated to the system increases, the efficiency $\mathcal{F}$ increases; however, the SFC system outperforms the other two.
In addition, the efficiency $\mathcal{F}$ of the \ac{SFC} system gets close to the optimal efficiency (i.e., $\mathcal{F}=1$) faster than the efficiency $\mathcal{F}$ of TDMA or slotted ALOHA. 
Note that the SFC system can outperform TDMA and slotted ALOHA even with a higher traffic load (i.e., three times higher $\lambda$).



The performance of the systems in terms of efficiency $\mathcal{F}$ versus $\lambda$ is shown in Fig. \ref{fig: Eff x lambda}. 
The parameters $R$, $k$, and $N$ are set to 11, 6, and 64, respectively.
For the range of values presented, whatever the traffic load (or $\lambda$) is, the \ac{SFC} system outperforms the other two in terms of efficiency $\mathcal{F}$.

Fig. \ref{fig: Eff x R and N} shows the system performance versus the number of resources $R$ with the ratio $N/R=6$ remaining constant.
Thereby, we can evaluate how systems behave when more users are admitted and mutually more resources are allocated to the system.
Again, the \ac{SFC} system outperforms TDMA and slotted ALOHA, and in addition, it manages to maintain a stable performance.
On the other hand, TDMA and slotted ALOHA experience a deterioration in efficiency $\mathcal{F}$ with increasing $R$ and $N$.
As the code word length increases (because of $k= \lceil \log_2 N \rceil$), the channel occupation time for transmitting one code word also increases.
Therefore, the probability of a collision in the TDMA and slotted ALOHA systems also increases.

\begin{figure}[t!]
    \centering
    \includegraphics[width=0.7\linewidth]{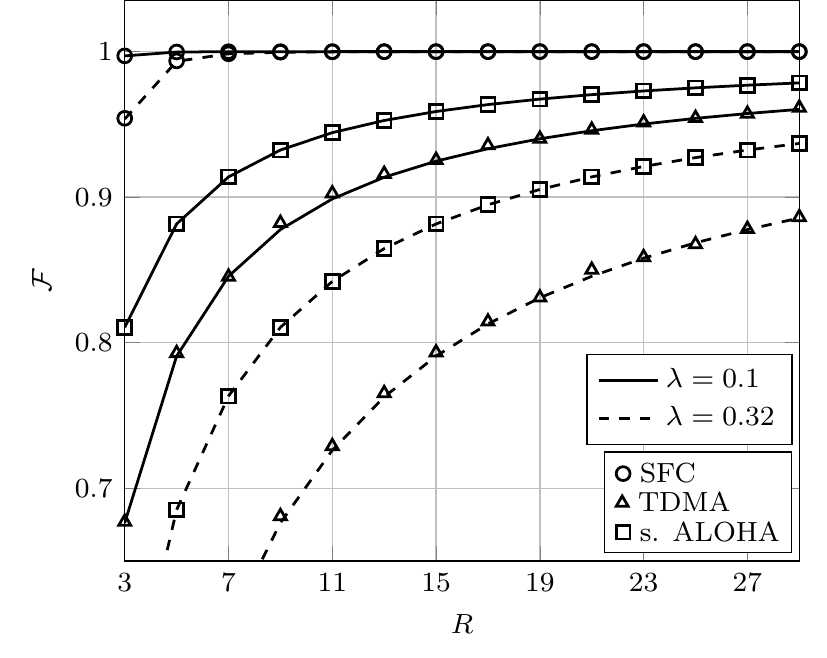}
    \caption{Efficiency $\mathcal{F}$ of TDMA, slotted ALOHA, and \ac{SFC} against $R$ for $k=6$, and $N=64$.}  
    \label{fig: Eff x R}
\end{figure}
\begin{figure}[t!]
    \centering
    \includegraphics[width=0.7\linewidth]{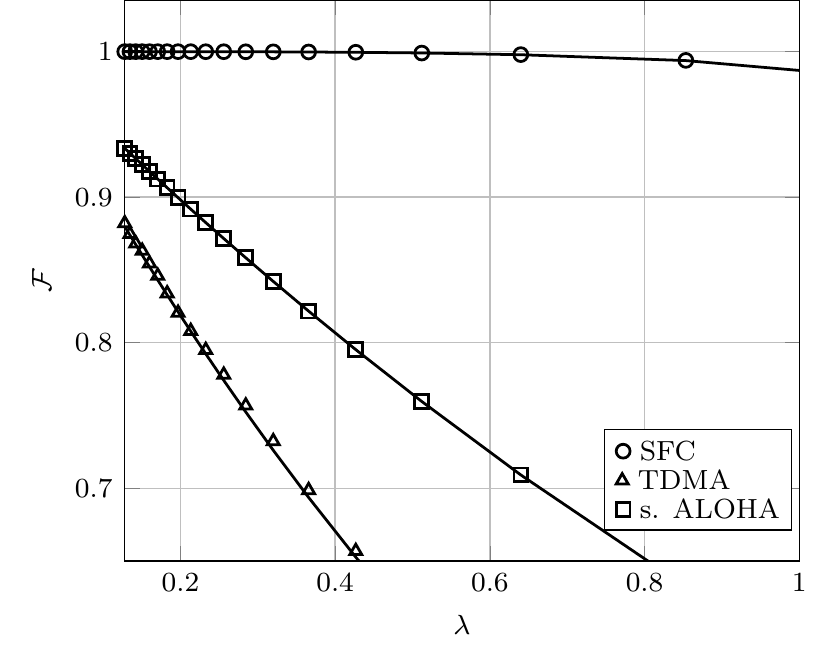}
    \caption{Efficiency $\mathcal{F}$ of TDMA, slotted ALOHA, and \ac{SFC} against $\lambda$ for $R=11$, $k=6$, and $N=64$.}
    \label{fig: Eff x lambda}
\end{figure}
\begin{figure}[t!]
    \centering
    \includegraphics[width=0.7\linewidth]{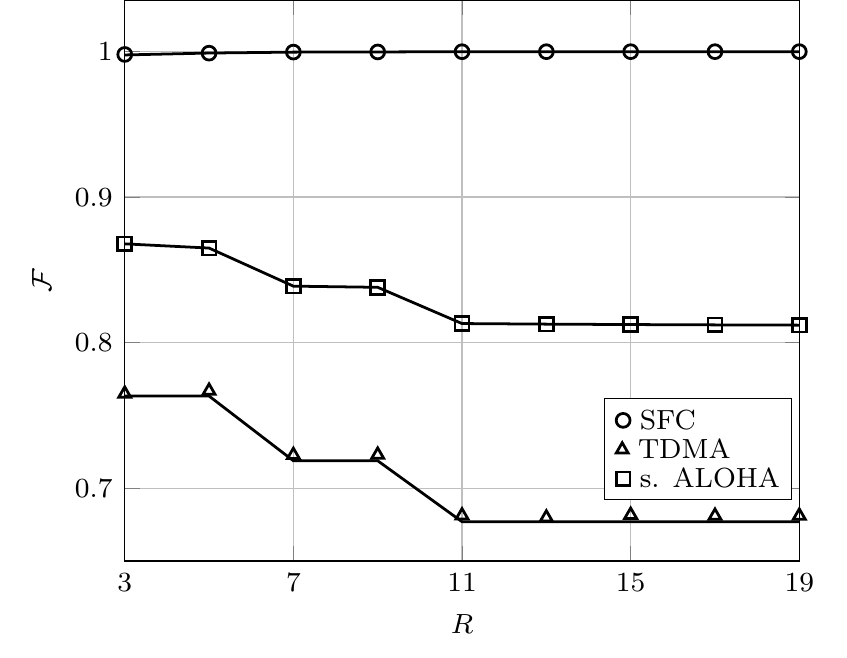}
    \caption{Efficiency $\mathcal{F}$ of TDMA, slotted ALOHA, and \ac{SFC} against $R$ for $N=6R$, and $\lambda=200/N$.}
    \label{fig: Eff x R and N}
\end{figure}

\begin{figure}[t!]
    \centering
    \includegraphics[width=0.7\linewidth]{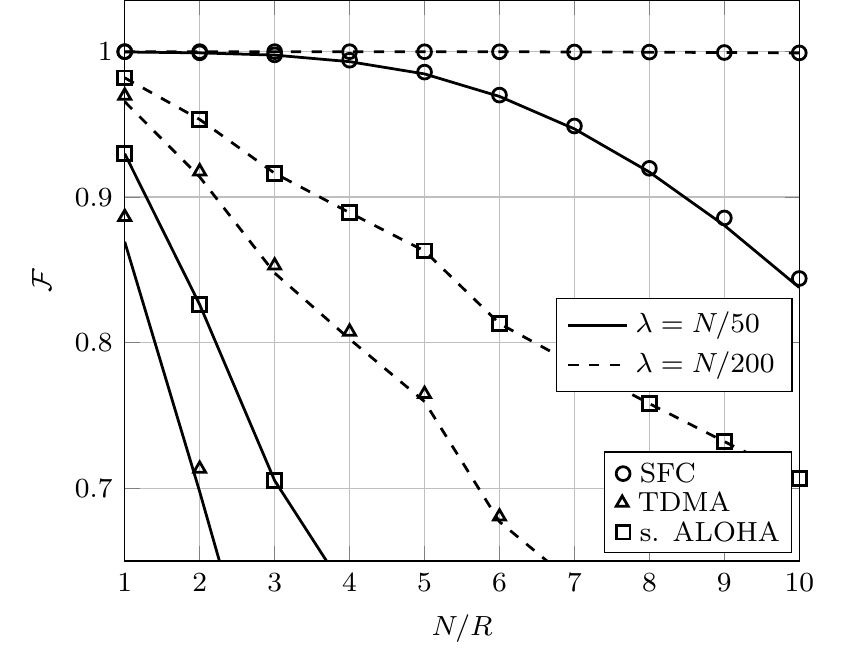}
    \caption{Efficiency $\mathcal{F}$ of TDMA, slotted ALOHA, and \ac{SFC} against $N/R$ for $R=11$.}
    \label{fig: Eff x N}
\end{figure}
\begin{figure}[t!]
    \centering
    \includegraphics[width=0.7\linewidth]{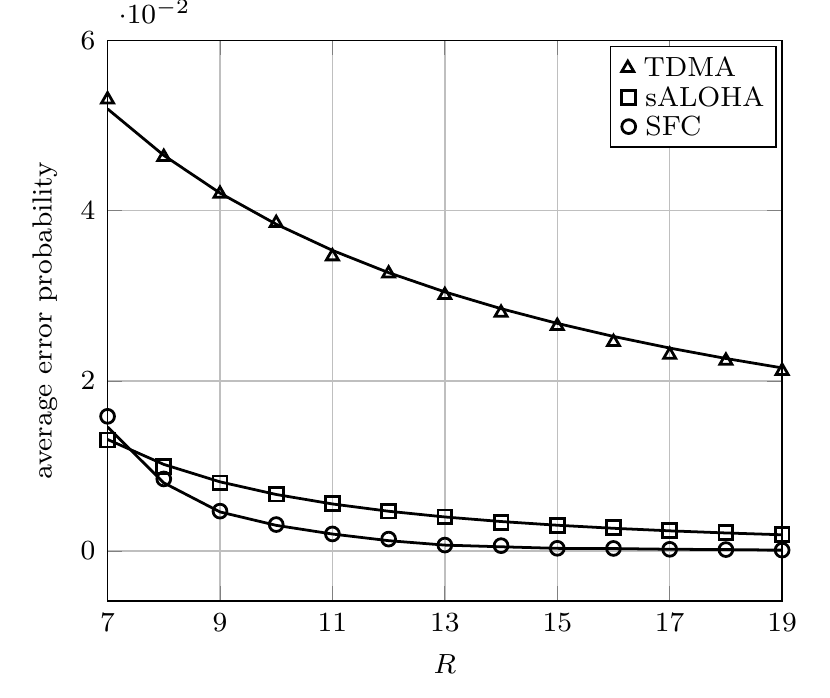}
    \caption{Average error probability of TDMA, slotted ALOHA, and \ac{SFC} against $R$ for $\lambda=0.2$, $k=6$, and $N=64$.}
    \label{fig: AEP x R}
\end{figure}

In Fig. \ref{fig: Eff x N}, the efficiency $\mathcal{F}$ for TDMA, slotted ALOHA, and \ac{SFC} is plotted against $N/R$.
Clearly, the \ac{SFC} system outperforms TDMA and slotted ALOHA. 
The behavior of the \ac{SFC} system is again remarkable, as in the entire operating region shown in Fig.  \ref{fig: Eff x N} the \ac{SFC} system is better than the TDMA and slotted ALOHA systems even with a higher traffic load (four times higher $\lambda$). 
Note that for $N/R=1$ the \ac{SFC} system has an efficiency $\mathcal{F}$ of 100\%, whereas TDMA and slotted ALOHA do not reach 100\% efficiency $\mathcal{F}$ in any of the simulated cases. 
This characteristic occurs owing to the possibility of two or more alarms occurring in the same sensor node within the period of a radio symbol.
When this happens, we have a collision for the TDMA and slotted ALOHA systems; however, the \ac{SFC} system can guarantee that sink receives the information correctly.

In Fig. \ref{fig: AEP x R}, the average error probability for \ac{TDMA}, slotted ALOHA, and \ac{SFC} is plotted against $R$ with $\lambda=0.2$, $k=6$, and $N=64$.
The markers represent simulated results, and the lines (analytical results) are obtained using \eqref{eq: prob error final}, \eqref{eq: prob error sALOHA}, and \eqref{eq: prob error TDMA}.
Clearly, the \ac{SFC} system outperforms TDMA and slotted ALOHA when $R>7$.
The TDMA system shows the worst performance.
Note how the error probability for the \ac{SFC} approaches 0 more quickly than for the other two schemes.
Again, it is evident how the SFC is able to more accurately estimate $\Theta$, notably in the case of a low resource allocation, i.e., $R<N$.



\section{Conclusion}\label{sec: Conclusion}
This paper introduced a novel approach to design communication systems, based on the meaning of the message to be transmitted and its end use.
We proposed a semantic-functional transmission of the state of sensors with respect to a predetermined event to be used in alarm systems.
Specifically, we proposed a random map to combine physical and MAC layers to constructively handle collisions with a low complexity.
Our numerical results demonstrated that the proposed \ac{SFC} achieves a transmission efficiency of 100\% for the proposed application in almost all the studied cases, outperforming the \ac{TDMA}- or slotted ALOHA-based systems in most of the scenarios evaluated.
This initial result will be extended to incorporate a more realistic environment that includes noise, fading, and other impairments of wireless communications, in which new types of mapping might be needed.
Besides, data from actual industrial IoT scenarios will be used to validate our framework when transmitting events related to fault detection \cite{dzaferagic2021fault}.

\section*{Acknowledgments}
This paper is partly supported by: Academy of Finland via (a) FIREMAN consortium n.326270 as part of CHIST-ERA grant CHIST-ERA-17-BDSI-003, (b) EnergyNet Fellowship n.321265/n.328869, (c) 6G Flagship program n. 346208, (d) Science Foundation Ireland under Grant number 13/RC/2077\_P2, (e) CNPq/Brazil (Grant Reference 311470/2021-1), (f) Fapemig (Grant No. APQ-00966-17), (g) São Paulo Research Foundation/Brazil (FAPESP) (Grant No. 2021/06946-0), (h) RNP, with resources from MCTIC/Brazil, Grant No. 01245.010604/2020-14, and (i) Brazil 6G project of the Radiocommunication Reference Center (\textit{Centro de Referência em Radiocomunicações} - CRR) of the National Institute of Telecommunications (\textit{Instituto Nacional de Telecomunicações} - Inatel), Brazil. We would like to thank Dr. Hanna Niemelä for proofreading this paper.

	\bibliographystyle{IEEEtran}
	\bibliography{IEEEabrv,References}

\end{document}

%% file: acronyms.tex
\DeclareAcronym{FDM}{
	short = FDM,
	long  = frequency-division multiplexing,
	tag = acrs,
}
\DeclareAcronym{OFDM}{
	short = OFDM,
	long  = orthogonal frequency-division multiplexing,
	tag = acrs,
}
\DeclareAcronym{OFDMA}{
	short = OFDMA,
	long  = orthogonal frequency-division multiple access,
	tag = acrs,
}
\DeclareAcronym{LTE}{
	short = LTE,
	long  = long term evolution,
	tag = acrs,
}
\DeclareAcronym{4G}{
	short = 4G,
	long  = fourth generation of mobile network,
	tag = acrs,
}
\DeclareAcronym{5G}{
	short = 5G,
	long  = fifth generation,
	tag = acrs,
}
\DeclareAcronym{SDMA}{
	short = SDMA,
	long  = space division multiple access,
	tag = acrs,
}
\DeclareAcronym{SDMA-OFDM}{
	short = SDMA-OFDM,
	long  = space division multiple access orthogonal frequency-division multiplexing,
	tag = acrs,
}
\DeclareAcronym{NOMA}{
	short = NOMA,
	long  = non-orthogonal multiple access,
	tag = acrs,
}
\DeclareAcronym{MIMO}{
	short = MIMO,
	long  = multiple-input-multiple-output,
	tag = acrs,
}
\DeclareAcronym{MIMO-NOMA}{
	short = MIMO-NOMA,
	long  = multiple-input-multiple-output non-orthogonal multiple access,
	tag = acrs,
}
\DeclareAcronym{CNN}{
	short = CNN,
	long  = convolutional neural network,
	tag = acrs,
}
\DeclareAcronym{LOS}{
	short = LOS,
	long  = line-of-sight,
	tag = acrs,
}
\DeclareAcronym{NLOS}{
	short = NLOS,
	long  = non-line-of-sight,
	tag = acrs,
}
\DeclareAcronym{WINNER}{
	short = WINNER,
	long  = Wireless World Initiative for New Radio,
	tag = acrs,
}
\DeclareAcronym{QuaDRiGa}{
	short = QuaDRiGa,
	long  = Quasi Deterministic Radio Channel Generator,
	tag = acrs,
}
\DeclareAcronym{MT}{
	short = MT,
	long  = mobile terminal,
	tag = acrs,
}
\DeclareAcronym{BS}{
	short = BS,
	long  = base station,
	tag = acrs,
}
\DeclareAcronym{NR}{
	short = NR,
	long  = new radio,
	tag = acrs,
}
\DeclareAcronym{AWGN}{
	short = AWGN,
	long  = additive white Gaussian noise,
	tag = acrs,
}
\DeclareAcronym{SNR}{
	short = SNR,
	long  = signal-to-noise ratio,
	tag = acrs,
}
\DeclareAcronym{ADAM}{
	short = ADAM,
	long  = Adaptive Moment Estimation,
	tag = acrs,
}
\DeclareAcronym{MSE}{
	short = MSE,
	long  = Mean Squared Error,
	tag = acrs,
}
\DeclareAcronym{MAC}{
	short = MAC,
	long  = Media Access Control,
	tag = acrs,
}
\DeclareAcronym{BER}{
	short = BER,
	long  = Bit Error Rate,
	tag = acrs,
}
\DeclareAcronym{WMSN}{
	short = WMSN,
	long  = Wireless Multimedia Sensor Network,
	tag = acrs,
}
\DeclareAcronym{WSN}{
	short = WSN,
	long  = Wireless Sensor Network,
	tag = acrs,
}
\DeclareAcronym{BTC}{
	short = BTC,
	long  = Brain-Type Communications,
	tag = acrs,
}
\DeclareAcronym{IoT}{
	short = IoT,
	long  = Internet of Things,
	tag = acrs,
}
\DeclareAcronym{PW-MAC}{
	short = PW-MAC,
	long  = Predictive-Wakeup MAC,
	tag = acrs,
}
\DeclareAcronym{TDMA}{
	short = TDMA,
	long  = Time-Division Multiple Access,
	tag = acrs,
}
\DeclareAcronym{FDMA}{
	short = FDMA,
	long  = Frequency-Division Multiple Access,
	tag = acrs,
}
\DeclareAcronym{CDMA}{
	short = CDMA,
	long  = Code-Division Multiple Access,
	tag = acrs,
}
\DeclareAcronym{BPSK}{
	short = BPSK,
	long  = binary phase-shift keying,
	tag = acrs,
}
\DeclareAcronym{QPSK}{
	short = QPSK,
	long  = quadrature phase-shift keying,
	tag = acrs,
}
\DeclareAcronym{MPSK}{
	short = MPSK,
	long  = $M$-ary phase-shift keying,
	tag = acrs,
}
\DeclareAcronym{RV}{
	short = RV,
	long  = random variables,
	tag = acrs,
}
\DeclareAcronym{pdf}{
	short = pdf,
	long  = probability density function,
	tag = acrs,
}
\DeclareAcronym{pmf}{
	short = pmf,
	long  = probability mass function,
	tag = acrs,
}
\DeclareAcronym{MAP}{
	short = MAP,
	long  = maximum a posteriori,
	tag = acrs,
}
\DeclareAcronym{MTE}{
	short = MTE,
	long  = meantime between events,
	tag = acrs,
}
\DeclareAcronym{SFC}{
	short = SFC,
	long  = semantic-functional communication,
	tag = acrs,
}
\DeclareAcronym{OOK}{
	short = OOK,
	long  = on-off keying,
	tag = acrs,
}